\documentclass[10pt,a4paper]{article}

\usepackage[utf8]{inputenc}
\usepackage{amsthm,amsmath,amssymb,amsfonts}
\usepackage{mathtools}
\usepackage{fullpage}
\usepackage{tikz}
\usepackage[hidelinks,draft=false,colorlinks=true]{hyperref}
\usepackage[capitalize]{cleveref}
\usepackage{dsfont}
\usepackage{authblk}
\usepackage{comment}

\usepackage[sorting=none]{biblatex}

\newtheorem{lemma}{Lemma}
\newtheorem{proposition}{Proposition}
\newtheorem{theorem}{Theorem}
\newtheorem{corollary}{Corollary}

\theoremstyle{definition}
\newtheorem{definition}{Definition}

\numberwithin{lemma}{section}
\numberwithin{proposition}{section}
\numberwithin{theorem}{section}
\numberwithin{definition}{section}
 
\newcommand{\tr}{\operatorname{Tr}}

\newcommand{\id}{\mathrm{Id}}
\newcommand{\C}{\mathbb{C}}

\newcommand{\End}{\mathrm{End}}
\newcommand{\ket}[1]{\vert #1 \rangle}
\newcommand{\bra}[1]{\langle #1 \vert}

\newcommand{\Span}{\mathrm{Span}}

\newcommand{\cH}{\mathcal{H}}
\newcommand{\cS}{\mathcal{S}}

\newcommand{\bA}{\mathbf{A}}
\newcommand{\bB}{\mathbf{B}}
\newcommand{\bC}{\mathbf{C}}
\newcommand{\cM}{\mathcal{M}}
\newcommand{\cV}{\mathcal{V}}

\usepackage{tcolorbox}
\definecolor{mycolor}{rgb}{0.122, 0.435, 0.698}
\newcommand{\mybox}[1]{%
  \begin{tcolorbox}[colframe=mycolor,boxrule=0.2pt,arc=4pt,left=4pt,right=4pt,top=4pt,bottom=4pt,boxsep=2pt] #1
  \end{tcolorbox}%
}

\usetikzlibrary{decorations.pathreplacing,decorations.markings,calc,3d,positioning,quotes}

\tikzset{
  tensor/.style={
    inner sep = 0.055cm,
    shape = circle,
    draw,
    fill
  },
  t/.style={
    inner sep = 0.03cm,
    shape = circle,
    draw,
    fill
  } 
    }

\title{MPS Stability and the Intersection Property}

\begin{document}

\author[1,2]{Jos\'e Garre-Rubio}
\author[3]{Alex Turzillo}
\author[1]{Andr\'as Moln\'ar}

\affil[1]{\small University of Vienna, Faculty of Mathematics, Oskar-Morgenstern-Platz 1, 1090 Vienna, Austria}
\affil[2]{\small Instituto de F\'isica Te\'orica, UAM/CSIC, C. Nicol\'as Cabrera 13-15, Cantoblanco, 28049 Madrid, Spain}
\affil[3]{\small Perimeter Institute for Theoretical Physics, 31 Caroline St N, Waterloo, ON N2L-2Y5, Canada}

\maketitle
\begin{abstract}
    We identify a property of the local tensors of matrix product states (MPS) that guarantees that their parent Hamiltonians satisfy the intersection property. The intersection property ensures that the ground space consists of MPS, with degeneracy bounded by the square of the bond dimension. The new local property, dubbed stability, generalizes (block) injectivity and is satisfied by the MPS tensors that construct the W state, domain wall superposition states, and their generalizations.
\end{abstract}

\section*{Introduction}

Matrix product states (MPS) \cite{PerezGarcia07} are a family of tensor network states that has shown success in a variety of applications to
strongly correlated quantum many-body systems, ranging from numerical calculations of phase diagrams in physics and quantum chemistry to analytical studies such as the phase classification of gapped symmetric spin chains \cite{Chen11,Schuch11}. MPS have also been extensively studied in machine learning \cite{stoudenmire17}.

By virtue of capturing the appropriate entanglement pattern -- an area law for the entanglement entropy -- MPS provide an efficient approximation of nondegenerate ground states of local and gapped one-dimensional Hamiltonians \cite{Hastings07A,Hastings07B}. A related property is that each MPS is associated with a local and gapped \emph{parent Hamiltonian} whose unique ground state is that MPS \cite{PerezGarcia07}. This correspondence between states and Hamiltonians provides a succinct formalism for studying one-dimensional models.

The correspondence, however, is limited to MPS with periodic boundary condition and does not cover states like the W state, which fall outside of this class but nevertheless can be described by MPS tensors with a nontrivial boundary condition (accordingly, the W state is not the unique ground state of a local and gapped Hamiltonian \cite{gioia2024}). To put it another way, the correspondence only holds for certain classes of MPS tensors: those which are injective \cite{PerezGarcia07,Molnar18A} or, more generally, block-injective (a.k.a. $G$-injective \cite{Schuch10}). As any periodic MPS can be brought into a canonical form \cite{PerezGarcia07}, it suffices to consider such tensors, whereas, for states like the W state, more general MPS tensors are necessarily involved.

Injective and block-injective MPS tensors satisfy a property called \emph{intersection}, which guarantees that the ground states of the parent Hamiltonian are precisely the corresponding MPS. The restriction to certain families of tensors and states is not just a technical assumption; it is a fundamental limitation: for generic tensors little can be said, and some questions about the ground space are undecidable \cite{Scarpa20}. Nevertheless, we find that the class of tensors for which the theory of intersection applies can be broadened. In this work we identify a property that generalizes block-injectivity and guarantees the intersection property. We dub this property \emph{stability}. This property is satisfied for important families of MPS tensors like those generating the W state, Dicke states, and domain wall superpositions.

We use stable MPS tensors to construct several families of MPS that satisfy the intersection property; for example, moving waves of an MPS tensor on a background of a different MPS, or superpositions of domain walls between different MPS. In each case, the ground space of the parent Hamiltonian contains not just the state of interest (e.g. the W state) but also a reference ground state (e.g. the state $|00\cdots\rangle$) above which the state of interest is an ``excitation.'' In light of a recent no-go theorem for realizing the W state as a unique ground state \cite{gioia2024}, this seems an unavoidable fact.

Both injectivity and stability are associated with a length which the interaction length of the parent Hamiltonian must be greater than in order for its space of ground states to be governed by the intersection property. At and below the stability length, there is no guarantee that the intersection property holds, and exotic behaviors such as exponential or polynomial growth of ground state degeneracy can occur \cite{Molnar25}. However, stability is only a sufficient but not necessary condition for intersection, and sometimes a parent Hamiltonian has a well-behaved ground space for anomalously low interaction length, as is the case for the AKLT model, which satisfies both stability and intersection on $2$ sites.

The structure of the paper is as follows. In \cref{sec:mps} we review the basics of MPS, their parent Hamiltonians, and the intersection property. In \cref{sec:stability} we introduce the concept of MPS stability, compare it to injectivity, and prove our main result -- that it implies intersection. In \cref{sec:examples} we discuss the aforementioned examples.

\section{MPS and their parent Hamiltonians}\label{sec:mps}


We consider two Hilbert spaces $\mathcal{H}$, called the \emph{physical space}, and $\mathbb{C}^D$, called the \emph{virtual space}, of finite dimension $d$ and $D$, respectively. Let $\{\ket{i}\in\mathcal{H} \mid i=0,\ldots, d-1\}$ form a basis of $\mathcal{H}$. An MPS tensor is an element $\bA = \sum_{i} A_i \otimes \ket{i}\in\cM_D\otimes \mathcal{H}$. The MPS with tensor $\bA$ on $k$ sites with boundary $X\in \mathcal{M}_D$ is the state
\begin{equation}
          \ket{X [\bA]^k } =\sum_{i_1,\cdots, i_k} \tr(X A_{i_1} \dots A_{i_{k}} )\,  \ket{i_1 \dots i_{k}}\in\mathcal{H}^{\otimes k}~.
\end{equation}
The \emph{physical MPS subspaces}
\begin{equation}
    \mathcal{S}_k(\bA) = \left\{ \ket{ X [\bA]^k } \, \Big | \, X\in \cM_D\right\}\subseteq\mathcal{H}^{\otimes k}
\end{equation}
will play a central role in our analysis. Notice that these spaces are invariant under conjugating the MPS tensor by an invertible matrix: $\mathcal{S}_k(\bA) = \mathcal{S}_k( T \bA T^{-1})$. In the following, we simply write $\cS_k$ when the MPS tensor $\bA$ is clear from context.

\begin{lemma}
    The physical MPS subspaces satisfy the relations
\begin{equation}\label{split}
    \cS_{k+\ell}\subseteq\cS_k\otimes\cS_\ell~,
\end{equation}
\begin{equation}\label{oneway}
    \cS_{k+1}\subseteq(\cS_k\otimes\cS_1)\cap(\cS_1\otimes\cS_k)=(\cS_k\otimes\cH)\cap(\cH\otimes\cS_k)~.
\end{equation}
\end{lemma}

\begin{proof}
Given an MPS state in $\cS_{k+\ell}$, write it as an element of $\cS_k\otimes\cS_\ell$:
\begin{equation}
    |X[\bA]^{k+\ell}\rangle=\sum_{\mu,\nu}|X[\bA]^ke_{\mu\nu}e_{\nu\mu}[\bA]^\ell\rangle = \sum_{\mu,\nu}|e_{\mu\nu}X[\bA]^k\rangle\otimes |e_{\nu\mu}[\bA]^\ell\rangle~,
\end{equation}
where $e_{\mu\nu}=|\mu\rangle\langle\nu|$ in a basis $\{|\mu\rangle\}$ of $\mathbb{C}^D$. This proves the inclusion \eqref{split}, which, in turn, straightforwardly implies the left-to-right inclusions \eqref{oneway}. To see the equality, note that $(\cS_k\otimes\cH)\cap(\cH\otimes\cS_k)\subseteq \cS_1\otimes\cS_{k-1}\otimes\cS_1$, and then use $(\cS_k\otimes\cH)\cap(\cS_1\otimes\cS_{k-1}\otimes\cS_1)=(\cS_k\otimes\cS_1)$ and $(\cH\otimes\cS_k)\cap(\cS_1\otimes\cS_{k-1}\otimes\cS_1)=(\cS_1\otimes\cS_k)$.
\end{proof}


An MPS tensor $\bA$ also has \emph{parent Hamiltonians}, for which the physical MPS subspace is (part of) the ground space, whenever this subspace is nonzero. For ease of notation, let us first define the cyclic translation operator $\tau_n$ on $n$ sites as the unitary operator that acts as $\tau_n \ket{i_1 \cdots i_n} = \ket{i_n i_1 \cdots i_{n-1}}$. 
\begin{definition}[Parent Hamiltonian]
    Let $\ell$, called the \emph{interaction length}, be such that $\mathcal{S}_\ell(\bA)$ is a proper subspace of $\mathcal{H}^{\otimes\ell}$, and let $h(\bA,\ell)$ be the orthogonal projector onto $\mathcal{S}_\ell(\bA)^{\perp}$. The $\ell$-local \emph{open boundary condition (OBC) parent Hamiltonian} of $\bA$ on $n$ sites is
    \begin{equation}
        H_n(\bA,\ell) = \sum_{i=0}^{n-\ell}\mathds{1}^{\otimes i} \otimes h(\bA,\ell) \otimes \mathds{1}^{ \otimes n-\ell-i}=\sum_{i=0}^{n-\ell} \tau_n^i(h(\bA,\ell) \otimes \mathds{1}^{ \otimes (n-\ell)})\tau_n^{-i}~,
    \end{equation}
    and the $\ell$-local \emph{periodic boundary condition (PBC) parent Hamiltonian} of $\bA$ on $n$ sites is
    \begin{equation}
        H^P_n(\bA,\ell) = \sum_{i=0}^{n-1} \tau_n^i(h(\bA,\ell) \otimes \mathds{1}^{ \otimes (n-\ell)})\tau_n^{-i}~.
    \end{equation}
\end{definition}
Note that the PBC Hamiltonian contains $\ell-1$ more terms than the OBC Hamiltonian. And unlike the OBC parent Hamiltonian, it is translation invariant:
\begin{equation}
    \tau_n H_n^P(\bA,\ell) = H_n^P(\bA,\ell) \tau_n~.
\end{equation}

Under the assumption that the MPS is nonvanishing, the parent Hamiltonians are \emph{frustration free}. Let us briefly review what this property means. Consider a Hamiltonian $H= \sum_i h_i$ such that the ground state energy of each term $h_i$ is non-negative: $h_i \ge 0$ and $\ker(h_i) \neq \{0\}$. Then, either of two cases holds:
\begin{enumerate}
    \item The ground state energy of $H$ is zero, in which case the ground space of $H$ is the intersection of the ground spaces of the $h_i$, i.e. $\cap_i GS(h_i) = GS(H)$, or
    \item The ground state energy is greater than zero and the intersection is trivial: $\cap_i GS(h_i) = \{0\}$.
\end{enumerate}
To see why this is true, note that, as a sum of non-negative terms $\langle\psi|h_i|\psi\rangle\ge 0$, the energy $\langle\psi|H|\psi\rangle$ of any state $|\psi\rangle$ is non-negative: $H\ge 0$. It vanishes if and only if each of these non-negative terms does. Thus $\ker H=\cap_i\ker h_i$. If the ground state energy $E_0$ vanishes, a state $|\psi\rangle\ne 0$ is a ground state of $H$ if and only if it is a ground state of each $h_i$. If $E_0>0$ then $\ker H=\cap_i\ker h_i$ contains only the zero vector. In the first case, the Hamiltonian $H$ is said to be \emph{frustration free} with respect to the terms $h_i$.

\begin{proposition}\label{parentFF}
    Suppose $\bA$ has nontrivial physical MPS subspace on $n$ sites: $S_n(\bA)\neq \{0\}$. Then, for any $\ell\le n$, the $\ell$-local OBC parent Hamiltonian on $n$ sites is frustration free with the MPS as ground states: 
    \begin{equation}
        \cS_n(\bA)\subseteq GS(H_n(\bA,\ell))~.
    \end{equation}
\end{proposition}

\begin{proof}
Repeatedly applying the inclusion \eqref{oneway} yields
\begin{equation}
    \cS_n(\bA)\subseteq\bigcap_{i=0}^{n-\ell}(\cH^{\otimes i}\otimes\cS_\ell(\bA)\otimes\cH^{\otimes n-\ell-i})=\bigcap_{i=0}^{n-\ell}\ker(\mathds{1}^{\otimes i} \otimes h(\bA,\ell) \otimes\mathds{1}^{ \otimes n-\ell-i})=\ker H_n(\bA,\ell)~.
\end{equation}
Since $\cS_n(\bA)$ is assumed to be nontrivial, this means the kernel is nontrivial, so it is the ground state space and the parent Hamiltonian $H_n(\bA,\ell)$ is frustration free.
\end{proof}

An analogous result can be stated for the PBC parent Hamiltonians. Let us first define the \emph{periodic MPS subspace} as the largest translation invariant subspace of the physical MPS subspace:
\begin{equation}
    \cS^P_n(\bA) = \bigcap_{i=0}^{n-1} \tau_n^i \cS_n(\bA) = \Span\{ \,\ket{\psi}\in \mathcal{S}_n(\bA) \,| \,\tau_n^i \ket{\psi}\in \mathcal{S}_n(\bA)~,\ i\in\{1, \ldots, n-1\}\,\}\subseteq\cS_n(\bA).
\end{equation}
Note that individual states in $\cS_n^P(\bA)$ need not be translation invariant. Then we have the result:

\begin{proposition}\label{pbc-parentFF}
    Suppose $\bA$ has nontrivial periodic MPS subspace on $n$ sites: $\cS_n^P(\bA)\ne\{0\}$. Then, for any $\ell\le n$, the $\ell$-local PBC parent Hamiltonian on $n$ sites is frustration free with the periodic MPS as ground states:
    \begin{equation}
        \cS_n^P(\bA)\subseteq GS(H^P_n(\bA,\ell))~.
    \end{equation}
\end{proposition}

\begin{proof}
    Note that $H^P_n(\bA,\ell)=\tfrac{1}{n-\ell+1}\sum_{i=0}^{n-1}\tau_n^i H_n(\bA,\ell)\tau_n^{-i}$, and therefore
    \begin{equation}
    \cS_n^P(\bA)=\bigcap_{i=0}^{n-1}\tau^i_n S_n(\bA)\subseteq \bigcap_{i=0}^{n-1}\ker\tau^i_n H_n(\bA,\ell)\tau_n^{-i}=\ker H^P_n(\bA,\ell)~.
    \end{equation}
    Since $\cS_n^P(\bA)$ is assumed to be nontrivial, these are frustration free ground states of $H_n^P(\bA,\ell)$.
\end{proof}

If the physical MPS subspace $\cS_n(\bA)$ or periodic MPS subspace $\cS_n^P(\bA)$ is trivial, the ground space of the parent Hamiltonian does not consist of MPS. If it is nontrivial, the ground space contains the MPS but might contain other states as well. If other ground states exist, the ground state degeneracy is not necessarily bounded as the system size $n$ grows. The following property of the MPS tensor, introduced in \cite{Schuch10}, guarantees that the ground space is precisely the space of MPS and nothing more:

\begin{definition}[Intersection Property]\label{intersection}
    An MPS tensor satisfies the intersection property on $k$ sites if
\begin{equation}
    \cS_{k+1}=(\cS_k\otimes\cH)\cap(\cH\otimes\cS_k)~,
\end{equation}
i.e., if the reverse of the inclusion \eqref{oneway} holds as well.
\end{definition}

\begin{proposition}\label{int-parent}
    Let $\bA$ be an MPS tensor that satisfies the intersection property for all $k$ with $\ell\leq k<n$. If $S_n(\bA)\ne \{0\}$, the ground space of the $\ell$-local OBC parent Hamiltonian at system size $n$ is exactly the physical MPS subspace:
    \begin{equation}
        GS(H_n(\bA,\ell))=\mathcal{S}_n(\bA)~.
    \end{equation}
    If $S_n^P(\bA)\ne\{0\}$, the ground space of the $\ell$-local PBC parent Hamiltonian at system size $n$ is exactly the periodic MPS subspace:
    \begin{equation}
        GS(H^P_n(\bA,\ell))=\cS_n^P(\bA)~.
    \end{equation}
    In particular, the ground state degeneracy is bounded above by $D^2$, independent of the system size $n$.
\end{proposition}

\begin{proof}
    The proof is the same as in \cref{parentFF,pbc-parentFF} but with equality instead of inclusion.
\end{proof}

We emphasize that the assumption that the MPS tensor satisfies intersection for \emph{all} $k$ with $\ell\leq k<n$ is crucial to \cref{int-parent}. One can construct examples \cite{Molnar25} where intersection holds for all $k\geq 3$, but not for $\ell = 2$, and consequently the $\ell=2$-body parent Hamiltonian on $n>2$ sites might have ground states beyond those in $\cS_n(\bA)$.

The following example shows that the intersection property on $k$ sites does not imply the intersection property on $k+1$ sites. Let $\bA$ be an MPS tensor that satisfies the intersection property for all $k\geq 2$ and has $\cS_k(\bA)\ne\{0\}$. Consider the MPS tensor $\bB$ defined by the matrices $B_i = A_i \otimes N$, where $N$ is a nilpotent matrix such that $N^4=0$. Then $\cS_k(\bB) = \cS_k(\bA)$ for all $k<4$, and thus $B$ satisfies intersection with $k=2$. However, $\cS_4(\bB) =0$, and thus $\bB$ does not satisfy intersection with $k=3$.

We see thus that conditions on the MPS tensor that guarantee the intersection property at all $k$ above a certain length are key for constructing a well-behaved parent Hamiltonian. In the following section we define such a condition and show that it generalizes previously known conditions (injectivity and $G$-injectivity). We illustrate that our condition is more general by a set of example in \cref{sec:examples} that fall outside of these previously studied classes.

\section{Stability: a sufficient condition for intersection}\label{sec:stability}

We now study conditions on the MPS tensor that guarantee the intersection property. To state them, we define the \emph{virtual MPS subspaces}
\begin{equation}
    \mathcal{V}_j(\bA)=\Span\{ \, A_{i_1} \dots A_{i_j} \, | \, i_1, \dots ,i_j \in \{0\dots d-1\} \,\}\subseteq\cM_D~.
\end{equation}
We will simply write $\cV_j$ when the MPS tensor $\bA$ is clear from context. Notice that these spaces may be concatenated according to
$\Span\{\cV_j\cdot\cV_k\}=\cV_{j+k}$. As the dimension of these spaces is at most $D^2$, where $D$ is the bond dimension, calculating them is fast: the construction of $\mathcal{V}_k$ takes $O(\log(k))$ time.


Let us start by reviewing injectivity and normality, before generalizing them below.

\begin{definition}[Injectivity]
    An MPS tensor $\bA$ is \emph{$j$-injective}, or \emph{injective at length $j$}, if
    \begin{equation}
        \cV_j(\bA)=\cM_D~.
    \end{equation}
    If $\bA$ is injective at some length $j$, we say it is \emph{normal} and the least such $j$ is called its \emph{injectivity length}.
\end{definition}

\begin{lemma}\label{injabovelength}
    Let $\bA$ be a normal MPS tensor with injectivity length $j$. Then $\bA$ is $k$-injective at all $k\geq j$.
\end{lemma}

\begin{proof}
Suppose $\bA$ is $k$-injective. Then it is $k+1$-injective since
\begin{equation}
    \cM_D=\cV_k=\Span\{\cV_1\cV_{k-1}\}\subseteq\Span\{\cV_1\cM_D\}=\Span\{\cV_1\cV_k\}=\cV_{k+1}
\end{equation}
means $\cV_{k+1}=\cM_D$. By induction, $\bA$ is injective for all $k\ge j$.
\end{proof}

In Ref. \cite{Schuch10} it was shown that a normal MPS tensor satisfies intersection above its injectivity length; below we recover this result as a special case of \cref{PropIP}. MPS tensors consisting of blocks of normal MPS tensors are also known to satisfy intersection; our results below generalize this case too.

Our new condition generalizes normal and block-injective MPS and reads as follows:
\mybox{
\begin{definition}[Stability]\label{def:stable}
An MPS tensor $\bA$ is said to be \emph{left $j$-stable}, or \emph{left stable at length $j$}, if there are matrices $Y_i\in \mathcal{M}_D$, for $i\in \{0, \dots , d-1\}$, such that
\begin{itemize}
    \item $Y_i\mathcal{V}_{j+1}\subseteq \mathcal{V}_j$, for all $i$, and 
    \item $Z = \sum_i A_iY_i$ is such that $ZV = V$, $\forall V \in \mathcal{V}_{j+1}$.
\end{itemize}
Similarly, $\bA$ is said to be \emph{right $j$-stable} if there are matrices $Y_i\in \mathcal{M}_D$ such that
\begin{itemize}
    \item $\mathcal{V}_{j+1}Y_i \subseteq \mathcal{V}_j$, for all $i$, and
    \item $Z = \sum_i Y_i A_i$ is such that $VZ = V$, $\forall V \in \mathcal{V}_{j+1}$.
\end{itemize}
\end{definition}
}

Notice that the stability property is preserved under conjugation of the tensor: $\bA$ is stable if and only if  $T\bA T^{-1}$ is stable. When left versus right is clear from context or when a statement holds either way, we simply say $\bA$ is \emph{stable}.

\begin{lemma}\label{injective-stable}
    A $j$-injective MPS tensor is stable at length $j$.
\end{lemma}

\begin{proof}
    Let $\bA$ be $j$-injective: $\cV_j=\cM_D$. Thus the condition $Y_i\cV_{j+1}\subseteq\cV_j$ is satisfied, for any choice of $Y_i$. The other condition is satisfied if the matrices $Y_i$ are such that $\sum_iA_iY_i=\mathds{1}$, i.e. they form an inverse to the map $\bA:\mathbb{C}^D\rightarrow\mathbb{C}^D\otimes\cH$. Normality guarantees that this inverse exists: if it did not, there would be a vector in the kernel of $\cV_1^T$, and so also in the kernel of $\cV_j^T$, in contradiction with $\cV_j=\cM_D$. This completes the proof of left stability. A similar argument shows right stability.
\end{proof}

Like injectivity, stability is preserved under concatenation:

\begin{lemma}\label{prop:stability_grow}
  If an MPS tensor is left (right) $j$-stable, it is left (right) stable at all lengths $k\geq j$.
\end{lemma}

\begin{proof}
Let $\bA$ be a left $j$-stable MPS tensor, and let $k\ge j$. Then
\begin{equation}
    Y_i\cV_{k+1}=Y_i\cdot\Span\{\cV_{j+1}\cV_{k-j}\}=\Span\{Y_i\cV_{j+1}\cV_{k-j}\}\subseteq\Span\{\cV_j\cV_{k-j}\}=\cV_k~.
\end{equation}
Since $\cV_{k+1}=\Span\{\cV_{j+1}\cV_{k-j}\}$, write $V\in\cV_{k+1}$ as $V=\sum_\mu U_\mu W_\mu$ for $U\in\cV_{j+1}$ and $W_\mu\in \cV_{k-j}$. Then
\begin{equation}
    ZV=\sum_\mu ZU_\mu W_\mu=\sum_\mu U_\mu W_\mu=V~.
\end{equation}
A similar argument holds for right stable MPS tensors.
\end{proof}

Before proving our main result in \cref{PropIP}, we demonstrate an intermediate result. The following \cref{prop:Opmove} shows on how a left (right) stable tensor implies the existence of a physical operator that pushes the virtual boundaries to the left (right):

\mybox{
\begin{proposition}\label{prop:Opmove}
    Let $\bA$ be a left $j$-stable MPS tensor. There is an operator $O\in\End(\mathcal{H}^{\otimes j+1})$ such that $O \cdot [\bA]^{j+1} = [\bA]^{j+1}$ and, for every $M \in \mathcal{M}_D$, a matrix $N \in \mathcal{M}_D$ with $ O \cdot \left ( \bA M [\bA]^j\right ) = N [\bA]^{j+1}$:
    \begin{equation}\label{optrans}
    \begin{tikzpicture}[baseline=1pt]
    \foreach \x in {0,1,3}{
    \node[tensor] at (\x,0){};
    \draw[thick]{} (\x-0.5,0)--(\x+0.5,0);
    \draw[thick]{} (\x,0)--(\x,0.9);
    }
    \node[] at (2,0) {$\cdots$};
    \node[tensor,blue,label=below:$M$] at (0.5,0) {};
    \draw[rounded corners, fill=gray, opacity=1]  (-0.2,0.3) rectangle (3.2,0.7);
    \node[] at (1.5,0.5) {$O$};
    \end{tikzpicture}
    = 
    \begin{tikzpicture}[baseline=1pt]
    \foreach \x in {0,1,3}{
    \node[tensor] at (\x,0){};
    \draw[thick]{} (\x-0.5,0)--(\x+0.5,0);
    \draw[thick]{} (\x,0)--(\x,0.3);
    }
    \node[] at (2,0) {$\cdots$};
    \draw[thick]{} (-1,0)--(-0.5,0);
    \node[tensor,red,label=below:$N$] at (-0.5,0) {};
    \draw [decorate,decoration={brace,amplitude=5pt,raise=1ex}] (-0.1,0.2) -- (3.1,0.2) node[midway,yshift=1.5em]{\small $j+1$};
    \end{tikzpicture}\  .
    \end{equation}
\end{proposition}
}

\begin{proof}
Given a left $j$-stable MPS tensor $\bA$, we construct $O$ explicitly. Since $Y_i\cV_{j+1}\subseteq\cV_j$, we can write
\begin{equation}
    Y_iA_{i_0}\cdots A_{i_j}=\sum_{i_1',\ldots, i_j'}O(i)_{i_0\cdots i_j}^{i_1'\ldots i_j'}A_{i_1'}\cdots A_{i_j'}~.
\end{equation}
Interpret these coefficients as matrix elements to construct the operator
\begin{equation}
    O=\sum_{\substack{i_0,\ldots,i_j \\ i,i_1',\ldots,i_j'}}O(i)_{i_0\cdots i_j}^{i_1'\cdots i_j'}|i_0,\cdots, i_j\rangle\langle i,i_1',\ldots, i_j'| : \mathcal{H}^{\otimes j+1} \mapsto \mathcal{H}^{\otimes j+1}~ .
\end{equation}
Then we can check  $ O \cdot \left ( \bA M [\bA]^j\right ) = N [\bA]^{j+1}$:
\begin{align}\begin{split}
    O\cdot\sum_{i_0',\ldots,i_j'}A_{i_0'}MA_{i_1'}\cdots A_{i_j'}|i_0'\cdots i_j'\rangle
    &=\sum_{i_0,\ldots,i_j}\sum_{i_0',\ldots,i_j'}O(i_0')^{i_1'\cdots i_j'}_{i_0\cdots i_j}A_{i_0'}MA_{i_1'}\cdots A_{i_j'}|i_0\cdots i_j\rangle\\
    &=\sum_{i_0,\ldots,i_j}\left(\sum_{i_0'}A_{i_0'}MY_{i_0'}\right)A_{i_0}A_{i_1}\cdots A_{i_j}|i_0\cdots i_j\rangle~,
\end{split}\end{align}
so $N=\sum_iA_iMY_i$. In particular, the choice $M=\mathds{1}$ results in $N=\sum_iA_iY_i=Z$. Then from the stability condition $ZV=V$ for all $V\in\cV_{j+1}$, it follows that $O\cdot[\bA]^{j+1}=[\bA]^{j+1}$.
\end{proof}

An analogous result holds for right $j$-stable MPS tensors, where $O$ acts as $ O \cdot \left ( [\bA]^j M {\bA}\right ) =   [\bA]^{j+1} N $. The proof of this statement follows the same logic as the left version. For normal tensors with injectivity length $j$ an operator that satisfies Eq.~\eqref{optrans} is $O = \sum_{i,i'} \tr \{ A_{i_0}\cdots A_{i_j}A^{-1}_{i'_j\cdots i'_1} Y_{i'_0}\} \ket{i}\bra{i'}$ where $A^{-1}$ is the inverse of $j$ A tensors and $Y$ is the tensor satisfying $\sum A_iY_i=\id$, see \cref{injective-stable}. Beyond the normal case, we leave an explicit formula for $O$ in terms of the stability matrices $Y_i$ for future work.

Finally we state our main result:
\mybox{
\begin{theorem}\label{PropIP}
If an MPS tensor $\bA$ is stable at length $j$, then it also satisfies the intersection property at all lengths $k\geq j+1$:
\begin{equation}
 (\mathcal{S}_k\otimes \mathcal{H}) \cap (\mathcal{H}\otimes \mathcal{S}_k) = \mathcal{S}_{k+1} \ .    
\end{equation}

\end{theorem}
}

\begin{proof} Due to \cref{prop:stability_grow}, it is enough to consider $k=j+1$. By Eq. \eqref{oneway}, a state $\ket{\Psi}$ in the intersection lies in both $\cS_k\otimes\cS_1$ and $\cS_1\otimes\cS_k$. Let $O\in \End(\mathcal{H}^{\otimes k})$ be an operator as in \cref{prop:Opmove}. Then $\ket{\Psi} \in \cS_k\otimes\cS_1$ means that the state is fixed by the action of $O\otimes\mathds{1}$, since $O\cdot[\bA]^k=[\bA]^k$. On the other hand, $\ket{\Psi}\in \cS_1\otimes\cS_k$ means that we can write the state as $\ket{\Psi} = \sum_q \ket{X_q\bA} \ket{X_q'[\bA]^k}$, and thus, as there are matrices $M_p$ and $V_p$ such that the equality
\begin{equation}\label{eq:equicdec}
    \sum_q \ 
    \begin{tikzpicture}[baseline=-1mm, yscale=0.5, font=\scriptsize, xscale = 0.7]
        \draw (-0.7, 0) rectangle (0.7,-1); 
        \draw (2.3,0) rectangle (5.7,-1);
        \foreach \x in {0,3,5}{
            \node[tensor] (t) at (\x,0) {};
            \node at (\x,-0.5) {$A$}; 
            \draw (t)--++(0,0.7);
        }
        \node[fill=white] at (4,0) {$\ldots$};
        \node[tensor] at (0.7,-0.5) {};
        \node at (1.2,-0.5) {$X_q$};
        \node[tensor] at (2.3,-0.5) {};
        \node at (1.8,-0.5) {$X_q'$};
    \end{tikzpicture} =  
    \sum_p \ 
    \begin{tikzpicture}[baseline=-1mm, yscale=0.5, font=\scriptsize, xscale = 0.7]
        \draw (-0.7, 0) rectangle (4.7,-1); 
        \foreach \x in {0,2,4}{
            \node[tensor] (t) at (\x,0) {};
            \node at (\x,-0.5) {$A$}; 
            \draw (t)--++(0,0.7);
        }
        \node[fill=white] at (3,0) {$\ldots$};
        \node[tensor] at (1,0) {};
        \node at (1,0.5) {$M_p$};
        \node[tensor] at (1,-1) {};
        \node at (1,-0.5) {$V_p$};
    \end{tikzpicture} 
\end{equation}
holds, $\ket{\Psi} = \sum_p|V_p[\bA]M_p[\bA]^k\rangle$. Applying $O\otimes\mathds{1}$ to this expression and using \cref{prop:Opmove}, we find
\begin{equation}
    \ket{\Psi} = \sum_p|V_p[\bA]M_p[\bA]^k\rangle=(O\otimes\mathds{1})\cdot\sum_p|V_p[\bA]M_p[\bA]^k\rangle=\sum_p|V_pN_p[\bA]^{k+1}\rangle\in\cS_{k+1}
\end{equation}
for some matrices $N_p$.
\end{proof}

Note that \cref{PropIP} together with \cref{injective-stable} recovers the previously known result that normal MPS tensors have the intersection property as a special case. A $j$-injective MPS tensor is $j$-stable and therefore satisfies intersection at all lengths $k\ge j+1$.


\cref{PropIP} can be interpreted in terms of parent Hamiltonians, using \cref{int-parent}. Let us first state one more condition on the local tensor:
\begin{definition}[Nilpotent MPS tensor]
    An MPS tensor $\bA$ is said to be nilpotent if its matrices generate a nilpotent algebra, i.e. if there is a length $n$ for which $\cV_n(\bA)=\{0\}$.
\end{definition}
Note that a non-nilpotent MPS tensor has $S_n(\bA)\ne\{0\}$ at all $n$ and therefore meets the assumptions of \cref{parentFF}; thus, its OBC parent Hamiltonian is frustration free. Therefore we have the statement

\mybox{
\begin{corollary}\label{corol}
    Let $\bA$ be an MPS tensor that is (left or right) stable at length $j$, and let $\ell\ge j+1$. If $\bA$ is non-nilpotent, the ground space of the $\ell$-local OBC parent Hamiltonian at system size $n$ is exactly the physical MPS subspace:
    \begin{equation}
        GS(H_n(\bA,\ell))=\mathcal{S}_n(\bA)~.
    \end{equation}
    If $\cS_n^P(\bA)\ne\{0\}$, the ground space of the $\ell$-local PBC parent Hamiltonian at system size $n$ is exactly the periodic MPS subspace:
    \begin{equation}
        GS(H^P_n(\bA,\ell))=\cS_n^P(\bA)~.
    \end{equation}
    In particular, the ground state degeneracies are bounded above by $D^2$.
\end{corollary}
}


\section{Examples}\label{sec:examples}

In this section we study tensors that satisfy the stability condition but are not injective. The first thing we analyze is how these tensors diverge from the block-injective form of Ref.~\cite{PerezGarcia07}.

\subsection{Boundary conditions and canonical form}\label{sec:canonical}

Let us analyze the form of the MPS wavefunctions $|X[\bA]^k\rangle\in\cS_n(\bA)$, which we have seen to appear as ground states of the MPS parent Hamiltonians.

For periodic boundary condition ($X=\mathds{1}$) MPS wavefunctions, there is a well-known \emph{canonical form theorem} \cite{PerezGarcia07}, which states that, for any MPS tensor $\bA$, there is another MPS tensor $\bB$ such that $|[\bA]^n\rangle=|[\bB]^n\rangle$ and $\bB$ is block-normal, i.e. $\bB=\oplus_\alpha\bB_\alpha$, where each block $\bB_\alpha$ is normal. As a consequence of this theorem, for the purposes of studying the PBC MPS it suffices to consider only normal MPS tensors.

If, however, one wishes to study MPS with more general boundary conditions $X$, the canonical form theorem does not apply, and so it is no longer the case that all phenomena are captured by normal MPS tensors. In the following subsections, it will be interesting to examine MPS wavefunctions of tensors that are stable but not normal. Before getting to examples, let us discuss some general theory.

Recall that the block-normal form of a PBC MPS is obtained by finding the irreducible subrepresentations of the MPS tensor, using a similarity transformation to make the MPS tensor block-triangular with respect to these subspaces, then throwing away the off-block-diagonal components \cite{PerezGarcia07}. For simplicity, suppose we have already performed the similarity transformation and that there are only two blocks. This means there is a projector $P$ satisfying $\bA P=P\bA P$, in terms of which the block-normal tensor is given by
\begin{equation}
    \bB=P\bA P+P^\perp\bA P^\perp~,
\end{equation}
where $P^\perp=\mathds{1}-P$. Note that this other projector satisfies $PP^\perp=P^\perp P=0$, $P^\perp\bA=P^\perp \bA P^\perp$ and $P^\perp\bA P=0$. Now consider a boundary condition $X$, and manipulate the corresponding MPS:
\begin{align}\begin{split}
    \ket{X[\bA]^n}
    &=\ket{X[(P + P^\perp)\bA(P + P^\perp)]^n}\\
    &=\ket{X[P\bA P]^n} + \ket{X[P^\perp\bA P^\perp]^n} + \sum_{\ell = 1}^{n-1}\ket{(P^\perp XP)[P\bA P]^{\ell-1}[P\bA P^\perp][P^\perp\bA P^\perp]^{n-\ell}}\\
    &=\ket{X[\bB]^n}+\sum_{\ell = 1}^{n-1}\ket{(P^\perp XP)[P\bA P]^{\ell-1}[P\bA P^\perp][P^\perp\bA P^\perp]^{n-\ell}}~.
\end{split}\end{align}
When $P^\perp XP=0$ (as in the PBC case) or when $P\bA P^\perp=0$ (i.e. if $\bA$ is already block-normal), the sum vanishes and one is left with the MPS for the canonical form tensor $\bB$. But in general, the sum remains. It is a moving wave superposition of MPS with tensor $P\bA P$ on the left, tensor $P\bA P^\perp$ at some position $\ell$ in the middle, and tensor $P^\perp \bA P^\perp$ on the right. This is the effect we see in the examples below.

\subsection{The W-state MPS tensor}
Let us consider the MPS tensor $\bA = A_0 \otimes \ket{0} + A_1\otimes\ket{1}$ specified by the matrices
\begin{equation}\label{Wstatensor}
    A_0=\mathds{1}_2,\quad  A_1 = \ket{0}\bra{1}. 
\end{equation}
This MPS tensor generates the W-state for the boundary condition $X= \frac{1}{\sqrt n}\sigma^+ = \frac{1}{\sqrt n}\ket{1}\bra{0}$:
\begin{equation}
     \ket{W_n}:=\frac{1}{\sqrt n}\left(|10\cdots 0\rangle+|01\cdots 0\rangle+\cdots+|00\cdots 1\rangle\right) 
     = \frac{1}{\sqrt n}\sum_{i_1,\ldots,i_n} \bra{0} A_{i_1}A_{i_2}\dots A_{i_n} \ket{1} \cdot \ket{i_1 i_2 \dots i_n}.
\end{equation}
In the language of \cref{sec:canonical}, this tensor has $P=|0\rangle\langle 0|$, which yields the block-normal tensor $\bB=\mathds{1}\otimes|0\rangle$. Since $|X[\bB]\rangle$ vanishes, the state is nothing but the moving wave superposition with $P\bA P=|0\rangle\langle 0|\otimes |0\rangle$ on the left, $P\bA P^\perp=|0\rangle\langle 1|\otimes |1\rangle$ in the middle, and $P^\perp\bA P^\perp=|1\rangle\langle 1|\otimes|0\rangle$ on the right.

Using the facts that $[A_0,A_1]=0$ and $A_1A_1=0$, one can compute the open and periodic MPS spaces:
\begin{equation}
    \cS_n^P(\bA)=\mathcal{S}_n(\bA) = \Span\{\,|0\rangle^{\otimes n}, |W_n\rangle\,\}~.
\end{equation}
Then we can determine the parent Hamiltonian ground spaces:
\mybox{
\begin{proposition}
The MPS tensor of the W-state satisfies intersection at lengths $k\ge 2$. The ground spaces of its OBC and PBC parent Hamiltonians coincide and are spanned by $|0\rangle^{\otimes n}$ and $|W_n\rangle$.
\end{proposition}
}
\begin{proof}
The W state MPS tensor defined in Eq.~\eqref{Wstatensor} is $1$-stable  (see \cref{def:stable}) with $Y_i=A_i$: indeed, $\mathcal{V}_j = \Span\{\mathds{1},\ket{0}\bra{1}\}$ for all $j$, and thus $Y_i \mathcal{V}_2 = A_i \mathcal{V}_2 \subseteq \mathcal{V}_3 = \mathcal{V}_1$ and $Y_0 A_0 + Y_1 A_1 = \mathds{1}$. Then by \cref{corol}, the ground spaces are given by $\cS_n(\bA)$ and $\cS_n^P(\bA)$, which we have just computed.
\end{proof}

An alternative proof may be given by direct calculation. Since $\mathcal{S}_{k+1} \subseteq (\mathcal{S}_k\otimes \C^2) \cap (\C^2\otimes \mathcal{S}_k)$ for free, we have to check the other inclusion. Consider an element in the intersection of 
\begin{equation}
\mathcal{S}_k(\bA)\otimes \C^2 = \Span \{ \ket{0}^{\otimes k+1}, \ket{0}^{\otimes k}\ket{1}, \ket{W_k}\ket{0},\ket{W_k}\ket{1}\},
\end{equation}
and
\begin{equation}
\C^2\otimes \mathcal{S}_k(\bA) = \Span \{ \ket{0}^{\otimes k+1}, \ket{1}\ket{0}^{\otimes k}, \ket{0}\ket{W_k},\ket{1}\ket{W_k}\}.   
\end{equation}
Since, for $k\ge 2$, there are terms in $\ket{W_k}\ket{1}$ that are not in $\ket{1}\ket{W_k}$, and vice-versa, these two factors are not in the intersection. Finally one has to realize that $\ket{0}^{\otimes k}\ket{1}+ \ket{W_k}\ket{0} = \ket{1}\ket{0}^{\otimes k}+ \ket{0}\ket{W_k} = \ket{W_{k+1}}$. 

For completeness, we write the local term of the $2$-local parent Hamiltonians explicitly:
\begin{equation}
    h= \ket{11}\bra{11}+ \frac{1}{2}(\ket{01}-\ket{10})(\bra{01}-\bra{10})= \begin{pmatrix}
    0 & 0 & 0 & 0 \\
    0 & 1/2 & -1/2 & 0 \\
    0 & -1/2 & 1/2 & 0 \\
    0 & 0 & 0 & 1
  \end{pmatrix} \ .
\end{equation}
A recent paper \cite{delcamp2024nonsemisimplenoninvertiblesymmetry} also discusses a local Hamiltonian with $|0\rangle^{\otimes n}$ and $|W\rangle$ as ground states.

Note that the $W$-state can be generalized by replacing $|0\rangle,|1\rangle$ with two nonzero vectors $\ket{a},\ket{b}\in\cH$:
\begin{equation}
    \bA=\mathds{1}_2\otimes \ket{a}+\ket{0}\bra{1}\otimes \ket{b} = \begin{pmatrix}
        \ket{a} & \ket{b} \\ 0 & \ket{a}
    \end{pmatrix}~.
\end{equation}
Up to normalization, this tensor gives the MPS
\begin{equation}\label{eq:W_generalization_a_b_MPS}
      \sum_{i_1,\ldots,i_n} \bra{0}A_{i_1}A_{i_2}\dots A_{i_n} \ket{1} \cdot \ket{i_1 i_2 \dots i_n}
      =|ba\cdots a\rangle+|ab\cdots a\rangle+\cdots+|aa\cdots b\rangle~.
\end{equation}
This tensor is $1$-stable: One can find $c_i$ and $d_i$ such that $\sum_i a_i c_i = 1$ and $\sum_i c_i b_i + d_i a_i=0$. Then the matrix $Y_i = c_i \mathds{1} + d_i \ket{0}\bra{1}$ satisfies $\sum_i Y_i A_i = \sum_i A_i Y_i = \mathds{1}$. Consequently, $Y_i V_{j+1} \subseteq V_j$ for all $j\geq 1$. The OBC and PBC ground spaces are spanned by the MPS in \cref{eq:W_generalization_a_b_MPS} and $|a\rangle^{\otimes n}$.

\subsection{Dicke states}

Let us consider the MPS tensor $\bA = A_0 \otimes \ket{0} + A_1\otimes\ket{1}$ specified by the matrices
\begin{equation}\label{supDicke}
 A_0= \mathds{1}_D, \quad A_1 = \sum_{i=0}^{D-2}\ket{i}\bra{i+1}.     
\end{equation}

Similarly to the W state tensor, since $A_0$ and $A_1$ commute and $(A_1)^D=0$, we have, for $n\ge D-1$,
\begin{equation}\label{eq:dicke}
   \cS_n^P(\bA)=\cS_n(\bA) = \{\,\ket{0}^{\otimes n}, \ket{W_n^1},\ket{W^2_n},\ldots,\ket{W^{D-1}_n}\,\} , 
\end{equation}
where $\ket{W^p_n} = \bra{0}[\bA]^n \ket{p}$ is a Dicke state \cite{PhysRev.93.99,B_rtschi_2019} given by the equal weight superposition of all $n$-qubit states that have exactly $p$-many $|1\rangle$'s in them; for example, $|W_n^1\rangle$ is the W state.

This MPS tensor has stability length $j=D-1$ with $Y_0 = A_0 = \mathds{1}$ and $Y_1 = (A_1)^{D-1} = \ket{0}\bra{D-1}$. Indeed, notice first that for all $j\geq D-1$, the virtual subspace is $\mathcal{V}_j = \Span\{\,\mathds{1}, A_1, \dots , (A_1)^{D-1}\,\}$. Then $Y_0 \mathcal{V}_D = \mathcal{V}_D = \mathcal{V}_{D-1}$ and $Y_1 \mathcal{V}_D = \Span\{(A_1)^{D-1}\}\subseteq \mathcal{V}_{D-1}$, and $A_0 Y_0 + A_1 Y_1 = \mathds{1}$. Therefore, the MPS tensor defined in Eq.~\eqref{supDicke} satisfies the intersection property at lengths $k\ge D$, and the ground space of both the open boundary and periodic boundary $D$-local parent Hamiltonians are the Dicke states \eqref{eq:dicke}. 

The intersection property can also be shown by direct calculation just as for the W-state. For $k\ge D$, consider the subspaces
\begin{align}\begin{split}
    \mathcal{S}_k(\bA)\otimes \C^2 & = \Bigg\{ \ket{0}^{\otimes k+1}, \overbrace{\ket{0}^{\otimes k}\ket{1},\ket{W_k}\ket{0}}^{\ket{W_{k+1}}},
\overbrace{\ket{W_k}\ket{1}, \ket{W^2_k}\ket{0}}^{\ket{W^2_{k+1}}}, \ldots, 
\overbrace{\ket{W^{D-2}_k}\ket{1}, \ket{W^{D-1}_k}\ket{0}}^{\ket{W^{D-1}_{k+1}}}, \ket{W^{D-1}_k}\ket{1} \Bigg\}~,\\
\C^2 \otimes \mathcal{S}_k(\bA)  &= \Bigg\{ \ket{0}^{\otimes k+1}, 
\overbrace{ \ket{1}\ket{0}^{\otimes k} ,\ket{0}\ket{W_k}}^{\ket{W_{k+1}}},
\overbrace{ \ket{1}\ket{W_k}, \ket{0}\ket{W^2_k} }^{\ket{W^2_{k+1}}}, \ldots, 
\overbrace{\ket{1}\ket{W^{D-2}_k}, \ket{0}\ket{W^{D-1}_k}}^{\ket{W^{D-1}_{k+1}}}, \ket{1}\ket{W^{D-1}_k} \Bigg\} \ .
\end{split}\end{align}
Their intersection is $\mathcal{S}_{k+1}(\bA)$ since the terms under the brackets must be equal and (since $k\ge D$) the last term of each subspace does not belong to the intersection.

\subsection{W state with momentum}

Another generalization of the W state on $n$ sites is given by
\begin{equation}
    |W_n(p)\rangle=\frac{1}{\sqrt{n}}\sum_{j=1}^ne^{-ip(j-1)}X_j|0\rangle^{\otimes n}~,
\end{equation}
for some value $p$, and is described by the MPS tensor and boundary condition
\begin{equation}
    A_0=e^{-ip}|0\rangle\langle 0|+|1\rangle\langle 1|~,\quad A_1=|0\rangle\langle 1|~,\quad X=\frac{1}{\sqrt{n}}|1\rangle\langle 0|~.
\end{equation}
Using the relations $A_0A_1=e^{-ip}A_1A_0$ and $A_1A_1=0$, we find that the physical MPS subspace for this tensor is given by
\begin{equation}\label{WmomSn}
    \cS_n(\bA)=\Span\{\,|0\rangle^{\otimes n}, |W_n(p)\rangle\,\}~.
\end{equation}

When $p$ is quantized as $p\in (2\pi/n)\mathbb{Z}$, the state $|W_n(p)\rangle$ is an eigenstate of the translation operator $\tau_n$ with eigenvalue $e^{ip}$, so $p$ may be regarded as a momentum. In this case, the periodic MPS subspace coincides with the open MPS subspace. In general, since $\tau_nX_j=X_{j+1}\tau_n$, the translated state is
\begin{equation}
    \tau_n|W_n(p)\rangle=\frac{1}{\sqrt{n}}\sum_{j=1}^ne^{-ip(j-1)}X_{j+1}|0\rangle^{\otimes n}=\frac{e^{ip}}{\sqrt{n}}\left(e^{-ipn}X_1|0\rangle^{\otimes n}+\sum_{j=2}^ne^{-ip(j-1)}X_j|0\rangle^{\otimes n}\right)~,
\end{equation}
which only lies in $\cS_n(\bA)$ if $n=1$ or the $e^{-ipn}$ coefficient of the first term vanishes, i.e. if $p$ is quantized correctly. Therefore the periodic MPS subspace is
\begin{equation}\label{Wmompbc}
    \cS_n^P(\bA)=\left\{\begin{array}{ll}\Span\{|0\rangle^{\otimes n},|W_n(p)\rangle\}&\quad\text{if }p\in(2\pi/n)\mathbb{Z}\text{ or }n=1\\\Span\{|0\rangle^{\otimes n}\}&\quad\text{otherwise}\end{array}\right.~.
\end{equation}
The virtual subspaces are
\begin{equation}
    \cV_j(\bA)=\Span\{\,(A_0)^j, (A_0)^{j-1}A_1\,\}~,
\end{equation}
and the MPS tensor is $1$-stable with $Y_0=e^{ip}|0\rangle\langle 0|+|1\rangle\langle 1|$ and $Y_1=0$. Indeed, since $Y_0A_0=\mathds{1}$, we have $Y_0\cV_{j+1}=\cV_j$ and $Y_1\cV_{j+1}=\{0\}\subseteq\cV_j$ and $A_0Y_0+A_1Y_1=\mathds{1}$. This implies that intersection holds at $k\ge 2$ and the OBC and PBC parent Hamiltonians have ground spaces Eq.~\eqref{WmomSn} and Eq.~\eqref{Wmompbc}, respectively.

For completeness, write the local term of the $2$-local parent Hamiltonians explicitly:
\begin{equation}
    h=\mathds{1}-|00\rangle\langle00|-\frac{(e^{-ip}|01\rangle+|10\rangle)(e^{ip}\langle01|+\langle10|)}{(e^{ip}\langle01|+\langle 10|)(e^{-ip}|01\rangle+|10\rangle)}=\begin{pmatrix}
        0 & 0 & 0 & 0 \\
        0 & 1/2 & -e^{-ip}/2 & 0 \\
        0 & -e^{ip}/2 & 1/2 & 0 \\
        0 & 0 & 0 & 1        
    \end{pmatrix}~.
\end{equation}

\subsection{Superposition of domain walls}

The GHZ state on $n$ sites $|GHZ_n\rangle=|0\cdots0\rangle+|1\cdots 1\rangle$ is described by the $D=2$ block-normal MPS tensor $A_0=|0\rangle\langle 0|$, $A_1=|1\rangle\langle 1|$ with $X=\mathds{1}$. This tensor defines the space $\cS_n(\bA)=\Span\{|0\cdots0\rangle,|1\cdots 1\rangle\}$.

Let us consider the following tensor, obtained from the GHZ tensor by introducing a block-off-diagonal component:
\begin{equation}
    A_0=\begin{pmatrix}
        1&1\\0&0
    \end{pmatrix}~,\qquad A_1=\begin{pmatrix}
        0&0\\0&1
    \end{pmatrix}~.
\end{equation}
This tensor satisfies
\begin{equation}
    A_1A_1 = A_1~,\quad A_0A_0 = A_0~,\quad A_1A_0 = 0~,\quad\text{and}\quad A_0A_1=\begin{pmatrix}
        0&1\\0&0
    \end{pmatrix}~,
\end{equation}
from which we can check that, for $n\ge 2$,
\begin{equation}
    \mathcal{S}_n(\bA) = \left\{ \, \ket{0}^{\otimes n}, \ket{1}^{\otimes n}, \ket{DW_n} \, \right\},
\end{equation}
where 
\begin{equation}
    \ket{DW_n} = \sum_{i=2}^n \prod_{j=i}^n X^j \ket{0}^{\otimes n} = \ket{01\cdots 1}+\ket{001\cdots 1} + \cdots + \ket{0\cdots 01}
\end{equation}
is the superposition of all $n$-qubit states with one domain wall of type $(01)$. Like in the $W$ state example, a purely off-block-diagonal boundary condition $X=\sigma^+$ produces a moving wave superposition:
\begin{equation*}
    |\sigma^+[\bA]^n\rangle=|1\cdots 1\rangle+|DW_n\rangle+|0\cdots 0\rangle~.
\end{equation*}

The virtual MPS subspaces are given by
\begin{equation}
    \cV_1=\Span\{\,A_0, A_1\,\}~,\qquad\cV_{j\ge 2}=\Span\{\,A_0, A_1, A_0A_1\,\}~,
\end{equation}
and the stability property is satisfied at $j=2$ with, for example, $Y_0=A_0(1-A_1)$, $Y_1=A_1$, and $Z=\mathds{1}$. Therefore, the intersection property holds on $k\geq 3$ sites and thus the ground space of the $3$-local parent Hamiltonian on any system size $n\geq 3$ is $\mathcal{S}_n(\bA)$.

It is instructive to give an alternative proof of intersection by direct computation. First notice that
\begin{align}\begin{split}
     \mathcal{S}_k(\bA)\otimes \C^2 &= \Span\{ \,\ket{0}^{\otimes k+1}, \ket{1}^{\otimes k+1}, \ket{0}^{\otimes k}\ket{1}, \ket{DW_k}\ket{1},\ket{1}^{\otimes k}\ket{0}, \ket{DW_k}\ket{0} \,\}~,\\
     \C^2\otimes \mathcal{S}_k(\bA) & = \Span\{ \,\ket{0}^{\otimes k+1}, \ket{1}^{\otimes k+1}, \ket{0}\ket{1}^{\otimes k}, \ket{0}\ket{DW_k},\ket{1}\ket{0}^{\otimes k}, \ket{1}\ket{DW_k} \,\}~.
\end{split}\end{align} 
Therefore a state $v$ in the intersection can be written as 
\begin{align}\begin{split}
     v &= \alpha_1 \ket{0}^{\otimes k+1} +\alpha_2 \ket{1}^{\otimes k+1} + \alpha_3 \ket{0}^{\otimes k}\ket{1} + \alpha_4 \ket{DW_k}\ket{1} + \alpha_5 \ket{1}^{\otimes k}\ket{0}+ \alpha_6 \ket{DW_k}\ket{0} \\
       &= \beta_1  \ket{0}^{\otimes k+1}+\beta_2 \ket{1}^{\otimes k+1}+\beta_3 \ket{0}\ket{1}^{\otimes k} + \beta_4 \ket{0}\ket{DW_k} + \beta_5 \ket{1}\ket{0}^{\otimes k} + \beta_6 \ket{1}\ket{DW_k}. 
\end{split}\end{align} 
Comparing the two sides, we see that $\alpha_1 = \beta_1$, $\alpha_2 = \beta_2$, and $\alpha_3 = \alpha_4= \beta_3 = \beta_4$ and $\alpha_5 =\alpha_6 = \beta_5 = \beta_6 = 0$. Note that third relation requires $k\ge 3$, as $k=2$ only has $\alpha_3=\beta_4$ and $\alpha_4=\beta_3$ separately. Finally notice that $\ket{0}^{\otimes k}\ket{1} + \ket{DW_k}\ket{1} = \ket{DW_{k+1}}$ and thus $v\in \mathcal{S}_{k+1}(\bA)$. Therefore the intersection property holds for $k\geq 3$, and thus the $3$-local OBC ground space on $n$ sites is $\mathcal{S}_n(\bA)$. 

For completeness, we write the $3$-local parent Hamiltonian for this MPS tensor:
\begin{equation}
h = \mathds{1} - \ket{000}\bra{000} -\ket{111}\bra{111} - (\ket{001}+\ket{011})(\bra{001}+\bra{011})~.   
\end{equation}

This domain wall superposition on qubits can be generalized to any Hilbert space and two vectors $|a\rangle, |b\rangle\in\cH$. The MPS tensor
\begin{equation}
    \bA=\ket{0}\left(\bra{1}+\bra{0}\right)\otimes  \ket{a}+\ket{1}\bra{1}\otimes \ket{b}~,
\end{equation}
is $2$-stable and generates the subspace
\begin{equation}
    \mathcal{S}_n (\bA) = \{ \ket{a}^{\otimes n}, \ket{b}^{\otimes n}, \ket{DW^{ab}_n}\} \ , \quad \ket{DW^{ab}_n} = \ket{a}\ket{b}^{\otimes n-1}+\ket{a}^{\otimes 2}\ket{b}^{\otimes n-2} + \cdots + \ket{a}^{\otimes n-1}\ket{b}~.
\end{equation}

\subsection{Antiferromagnetic Ising model}

Consider the MPS tensor defined by matrices
\begin{equation}
  A_0= \ket{1}\bra{0}~, \quad \text{and} \quad A_1= \ket{0}\bra{1}~.  
\end{equation}
The MPS tensor describes the ground space of the antiferromagnetic Ising model; that is, 
\begin{equation}
   \mathcal{S}_n(\bA) = \Span\{ \ket{0101\cdots}, \ket{1010\cdots}\}.
\end{equation}
The $2$-local parent Hamiltonian has local term $h= \ket{11}\bra{11}+\ket{00}\bra{00}$.

This MPS tensor has virtual subspaces
\begin{equation}
    \mathcal{V}_\text{odd} = \Span \{ \,\ket{0}\bra{1}, \ket{1}\bra{0}\,\}~,\qquad\mathcal{V}_\text{even}= \Span \{ \,\ket{0}\bra{0}, \ket{1}\bra{1}\,\}~.
\end{equation}
It is $1$-stable with $Y_0=A_1$ and $Y_1=A_0$. Therefore the OBC ground space is $\cS_n(\bA)$.

The PBC ground state degeneracy depends on whether the system size $n$ is even or odd. For even $n$ the PBC ground space is the full $\mathcal{S}_n(\bA)$, i.e. $2$-fold degenerate. For odd $n$ the PBC parent Hamiltonian is frustrated as $\cS_n^P(\bA)$ is trivial, so our results relating the ground spaces to the physical MPS spaces do not apply. The degeneracy is $2n$; in particular, it is not bounded by the MPS bond dimension. There is a ground space basis in consisting of states that violate precisely one local Hamiltonian term. This basis is given by translating $1$ or $0$ defects on the patterns $(01)(01)...(01)$ or $(10)(10)...(10)$.

\subsection{Examples satisfying the intersection property}\label{sec:examples-nonstable}

Now we discuss three families of tensors, generalizing the W state and domain wall superpositions, that satisfy intersection despite not being stable. However, some of the tensors have some stable components.

\subsubsection{Generalized W states}

Now we define a family of MPS that generalizes the W state in the sense that they are superpositions of a moving tensor (the role of $1$) on a background of a different tensor (instead of $0$):
\begin{equation}
    \ket{W_n^{\bA, \bB}} = \sum_{\ell=1}^n\ket{ [\bA]^{\ell-1}[\bB][\bA]^{n-\ell}} \ .
\end{equation}
Given MPS tensors $\bA$ and $\bB$ with the same bond dimension $D$, the MPS tensor $\bC$ with bond dimension $2D$ given by
\begin{equation}\label{gen-W}
    C_i=\mathds{1}_2\otimes A_i+\ket{0}\bra{1}\otimes B_i = \begin{pmatrix}
        A_i & B_i \\ 0 & A_i
    \end{pmatrix}~.
\end{equation}
generates this generalized MPS state. One can compute the physical subspace
\begin{equation}
    \mathcal{S}_n(\bC) = \mathcal{S}_n(\bA) + \mathcal{S}^W_n(\bA, \bB) \ , \quad \mathcal{S}^W_n(\bA, \bB) = \Span\left\{ \, \sum_{\ell=1}^n\ket{ X [\bA]^{\ell-1}[\bB][\bA]^{n-\ell}} \, \middle| \, X\in \mathcal{M}_D \, \right\}~.
\end{equation}
We claim the following:
\begin{proposition}
    If $\bA$ satisfies the intersection property at length $k$ and $\mathcal{S}_1(\bA) \cap \mathcal{S}_1(\bB) = 0$, then $\bC$ also satisfies the intersection property at length $k$.
\end{proposition}

\begin{proof}
We want to show that $(\mathcal{S}_k(\bC)\otimes \mathcal{S}_1(\bC)) \cap (\mathcal{S}_1(\bC)\otimes \mathcal{S}_k(\bC)) = \mathcal{S}_{k+1}(\bC)$. The inclusion $\supseteq$ holds for any tensor, so we must show the inclusion $\subseteq$ for $\bC$. Note that
\begin{equation}
    \mathcal{S}_1(\bC) = \mathcal{S}_1(\bA) \oplus \mathcal{S}_1(\bB)~,
\end{equation}
and thus we are looking for the intersection of the following subspaces
\begin{align}\begin{split}
     \mathcal{S}_k(\bC)\otimes \mathcal{S}_1(\bC) =& \ \Span \{ \mathcal{S}_k (\bA)\otimes \mathcal{S}_1(\bA) ,
     \mathcal{S}_k(\bA)\otimes \mathcal{S}_1(\bB) ,
     \mathcal{S}^W_k(\bA, \bB)\otimes \mathcal{S}_1(\bA) ,
     \mathcal{S}^W_k(\bA, \bB)\otimes \mathcal{S}_1(\bB) \}
      \\
     \mathcal{S}_1(\bC)\otimes \mathcal{S}_k(\bC) =& \ \Span\{ \mathcal{S}_1(\bA)\otimes \mathcal{S}_k(\bA) , \mathcal{S}_1(\bA)\otimes \mathcal{S}^W_k(\bA, \bB) ,
     \mathcal{S}_1(\bB)\otimes \mathcal{S}_k(\bA) ,
     \mathcal{S}_1( \bB)\otimes \mathcal{S}^W_k(\bA, \bB)\}~.
\end{split}\end{align}
Since $\mathcal{S}_1(\bA) \cap \mathcal{S}_1(\bB) = 0 $, the intersection of the subspaces is a direct sum of intersections between linear combination of terms with the same number of $\bB$ tensors on them. With no $\bB$ tensors, we have the first term of each subspace. Due to the assumption that $\bA$ satisfies the intersection property, the intersection of the first term of each side is $\mathcal{S}_{k+1}(\bA)$. With two $\bB$ tensors, we have the last term of each subspace, and we can check their intersection is zero. To see this, note that if there is a nonzero element in the intersection there must exists nonzero matrices $X_1,X_2,Y_1,Y_2 \in \mathcal{M}_D$ such that
\begin{equation}
    \sum_{\ell=1}^n\ket{ X_1 [\bA]^{\ell-1}[\bB][\bA]^{n-\ell}} \otimes \ket{ Y_1 [\bB]} =  \sum_{\ell=1}^n \ket{ Y_2 [\bB]} \otimes \ket{ X_2 [\bA]^{\ell-1}[\bB][\bA]^{n-\ell}} \ .
\end{equation}
But this cannot happen since, for example, $\ket{ X_1 [\bA]^{n-1}[\bB]} \ket{ Y_1 [\bB]}$ does not appear in the superposition of the right hand side. Therefore, it remains to compute the intersection of the middle terms from each subspace, each of which contains a single $\bB$ tensor. We wish to show
\begin{equation}
    \Span\left \{ \mathcal{S}_k(\bA)\otimes \mathcal{S}_1(\bB),
     \mathcal{S}^W_k(\bA, \bB)\otimes \mathcal{S}_1(\bA) \right \} \cap \Span\left \{ \mathcal{S}_1(\bA)\otimes \mathcal{S}^W_k(\bA, \bB) ,
     \mathcal{S}_1(\bB)\otimes \mathcal{S}_k(\bA) \right \} = \mathcal{S}^W_{k+1}(\bA, \bB)~,
\end{equation}
where the inclusion $\supseteq$ is clear. A vector in the intersection can be written in two ways:
\begin{align}\begin{split}
    v = & \ \ket{X_1 [\bA]^k}\ket{X_2\bB} + \sum_{\ell=1}^k\ket{X_3 [\bA]^{\ell-1}[\bB][\bA]^{k-\ell}}\ket{X_4 \bA}~, \\
    v = & \ \sum_{\ell=1}^k \ket{X_5\bA} \ket{X_6[\bA]^{\ell-1}[\bB][\bA]^{k-\ell}} + \ket{X_7\bB}\ket{X_8[\bA]^k}\ .
\end{split}\end{align}
Since $\mathcal{S}_1(\bA) \cap \mathcal{S}_1(\bB) = 0 $, the components in the linear superposition of $v$ where the tensor $\bB$ is in the same position must coincide. This means, for example, that there must be $X_1,X_2,X_5,X_6$ such that $\ket{X_1 [\bA]^k}\ket{X_2\bB} = \ket{X_5\bA} \ket{X_6[\bA]^{k-1}[\bB]}$, which implies that this vector is proportional to $\ket{X[\bA]^k[\bB]}$ for some $X\in \mathcal{M}_D$. Arguing in the same way for each position of $\bB$, we conclude that $v$ is in $\mathcal{S}^W_{k+1}(\bA, \bB)$.
\end{proof}

\subsubsection{Generalized domain wall superpositions}

Given MPS tensors $\bA$ and $\bB$ with the same bond dimension, consider the MPS tensor $\bC$ given by
\begin{equation}\label{gen-domain}
    C_i=\ket{0}(\bra{1}+\bra{0})\otimes A_i+\ket{1}\bra{1}\otimes B_i = \begin{pmatrix}
        A_i & A_i \\ 0 & B_i
    \end{pmatrix}~.
\end{equation}
The physical subspace generated by $\bC$ is
\begin{equation}
    \mathcal{S}_n(\bC) = \mathcal{S}_n(\bA) + \mathcal{S}_n(\bB) + \mathcal{S}^{DW}_n(\bA, \bB) \ , \quad \mathcal{S}^{DW}_n(\bA, \bB) = \left\{ \, \sum_{\ell=1}^{n-1}\ket{ X [\bA]^\ell[\bB]^{n-\ell}} \, \middle| \, X\in \mathcal{M}_D \,\right\}~.
\end{equation}

\begin{proposition}\label{GenDW}
    If $\bA$ is left $j_A$-stable, $\bB$ is right $j_B$-stable, and $\mathcal{S}_1(\bA) \cap \mathcal{S}_1(\bB) = \{0\} $, then $\mathcal{S}_k(\bC)$ satisfies intersection for $k\ge j_A+j_B+1$, and the periodic subspace is $\cS_n^P(\bC)=\mathcal{S}_n(\bA)\oplus\mathcal{S}_n(\bB)$.
\end{proposition}

\begin{proof}
We want to show that $(\mathcal{S}_k(\bC)\otimes \mathcal{S}_1(\bC)) \cap (\mathcal{S}_1(\bC)\otimes \mathcal{S}_k(\bC)) = \mathcal{S}_{k+1}(\bC)$, for any $k \ge j_A+j_B+1$. The inclusion $\supseteq$ is trivial, we thus only have to show the inclusion $\subseteq$. Note that
\begin{equation}
    \mathcal{S}_1(\bC) = \mathcal{S}_1(\bA) \oplus \mathcal{S}_1(\bB)~,
\end{equation}
and thus we are looking for the intersection of the following subspaces
\begin{align}\begin{split}
     \mathcal{S}_k(\bC)\otimes \mathcal{S}_1(\bC)= &\ \Span \{ \mathcal{S}_k^{\bA}\otimes \mathcal{S}_1^{\bA},
     \mathcal{S}_k^{\bB}\otimes \mathcal{S}_1^{\bB}  ,
     \mathcal{S}_k^{\bA}\otimes \mathcal{S}_1^{\bB} ,
     \mathcal{S}_k^{\bA \bB}\otimes \mathcal{S}_1^{\bB}  ,
     \mathcal{S}_k^{\bB}\otimes \mathcal{S}_1^{\bA} ,
     \mathcal{S}_k^{\bA \bB}\otimes \mathcal{S}_1^{\bA} \}
      \\
     \mathcal{S}_1(\bC)\otimes \mathcal{S}_k(\bC) =& \ \Span\{ \mathcal{S}_1^{\bA}\otimes \mathcal{S}_k^{\bA},
     \mathcal{S}_1^{\bB}\otimes \mathcal{S}_k^{\bB}  ,
     \mathcal{S}_1^{\bA}\otimes \mathcal{S}_k^{\bB} ,
     \mathcal{S}_1^{\bA}\otimes \mathcal{S}_k^{\bA \bB}  ,
     \mathcal{S}_1^{\bB}\otimes \mathcal{S}_k^{\bA} ,
     \mathcal{S}_1^{ \bB}\otimes \mathcal{S}_k^{\bA \bB} \},
\end{split}\end{align}
where we have used the shorthand $\mathcal{S}_k^{\bA} \equiv \mathcal{S}_k(\bA), \mathcal{S}_k^{\bB} \equiv \mathcal{S}_k(\bB), \mathcal{S}_k^{\bA \bB} \equiv \mathcal{S}^{DW}_k(\bA, \bB)$. Since $\cS_1(\bA)\cap\cS_1(\bB)$ is trivial by assumption, the first two terms of each subspace are orthogonal to the rest, and, due to the stability of $\bA$ and $\bB$, their intersections are $\mathcal{S}_{k+1}(\bA)$ and $\mathcal{S}_{k+1}(\bB)$. Moreover, the terms $\mathcal{S}_k^{\bB}\otimes \mathcal{S}_1^{\bA} $ and $\mathcal{S}_k^{\bA \bB}\otimes \mathcal{S}_1^{\bA}$ of the first subspace have $[\bA]$ tensors after $[\bB]$ tensors, and since this pattern does not occur in the other subspace and $\cS_1(\bA)\cap\cS_1(\bB)= 0$, these terms are not in the intersection. The same happens for the terms $\mathcal{S}_1^{\bB}\otimes \mathcal{S}_k^{\bA}$ and $ \mathcal{S}_1^{ \bB}\otimes \mathcal{S}_k^{\bA \bB}$ of the second subspace, with $\bB$ tensors before $\bA$ tensors. Therefore, it remains to compute the intersection of the two middle terms from each subspace. We wish to show
\begin{equation}
    \Span \{\mathcal{S}_k^{\bA}\otimes \mathcal{S}_1^{\bB} , \mathcal{S}_k^{\bA \bB}\otimes \mathcal{S}_1^{\bB}\} \cap  \Span \{ \mathcal{S}_1^{\bA}\otimes \mathcal{S}_k^{\bB},\mathcal{S}_1^{\bA}\otimes \mathcal{S}_k^{\bA \bB}  \} \ = \cS_{k+1}^{\bA\bB}~,
\end{equation}
where the inclusion $\supseteq$ is clear. A vector in the intersection can be written in two ways:
\begin{align}\begin{split}  
    v &=\sum_{X_1,Y_1}\alpha_{X_1,Y_1}
     \ket{ X_1 [\bA]^k}\ket{Y_1 [\bB]} + \sum_{X_2,Y_2}\beta_{X_2,Y_2} \sum_{\ell=1}^{k-1}\ket{ X_2 [\bA]^\ell[\bB]^{k-\ell}}\ket{Y_2 [\bB]} \\
     &\qquad\qquad=
     \sum_i \sum_{\ell=1}^k\tr[[\bA]^\ell [\bB]^{k-\ell} Z_i [\bB] W_i]~, \\
    v &=
     \sum_{X_3,Y_3}\gamma_{X_3,Y_3}
     \ket{ X_3 [\bA]}\ket{Y_3 [\bB]^k} + \sum_{X_4,Y_4}\delta_{X_4,Y_4}\sum_{\ell=1}^{k-1}\ket{X_4 [\bA]}\ket{ Y_4 [\bA]^\ell[\bB]^{k-\ell}} \\
     &\qquad\qquad=
     \sum_j \sum_{\ell=1}^k\tr[ V_j[\bA] U_j [\bA]^{\ell-1}[\bB]^{k-\ell+1} ]  \ ,
\end{split}\end{align}
where in the last equalities we have used the decomposition of Eq.~\eqref{eq:equicdec}. Using again that $\mathcal{S}_1(\bA) \cap \mathcal{S}_1(\bB) = 0$, we can compare terms with the same $\ell$. If $k\ge j_A+j_B+1$, then any $\ell\le k$ satisfies either $\ell-1\ge j_A$ or $k-\ell\ge j_B$ (or both). If $\ell-1\ge j_A$, since $\bA$ is left $j_A$-stable, we can use the operator $O$ of \cref{prop:Opmove} to move $U_j$ to the left of the chain. Similarly, if $k-\ell\ge j_B$, since $\bB$ is right $j_B$-stable, we can move $Z_i$ to the right of the chain. Either way $v$ has the form of a state in $\mathcal{S}^{DW}_{k+1}(\bA, \bB)$.

\end{proof}

We note that the family of tensors of the form
\begin{equation}    \begin{pmatrix}A_i&B_i\\0&C_i
\end{pmatrix}
\end{equation}
results in superpositions of domain walls between $\bA$ and $\bC$ with a defect tensor $\bB$ between them, as
\begin{equation}
    \begin{pmatrix}A_i&B_i\\0&C_i
    \end{pmatrix}^n =\begin{pmatrix}A^n&\sum_{0\le p \le n-1}A^{n-1-p}B C^p\\0&C^n\end{pmatrix}~.
\end{equation}
Therefore, this family includes the generalized W state tensor \eqref{gen-W} (taking $\bC=\bA$) and the domain wall superposition tensor \eqref{gen-domain} (taking $\bC=\bB$). Superpositions of domain walls between $\bA \in \mathcal{M}_{D_A} \otimes \mathcal{H}$ and $\bB \in \mathcal{M}_{D_B} \otimes \mathcal{H}$ with virtual defects $W \in \mathcal{M}_{D_A \times D_B}$ in between are obtained by $\bB=W\tilde{\bB}, \bC=\tilde{\bB}$.

\subsection{Generalized intersection property}
\label{counterex}

Finally we discuss an example the fails the intersection property yet satisfies a more general condition that also ensures the ground space is well-behaved.

Consider the MPS tensor
\begin{equation}
    A_0 = \begin{pmatrix} 1 & 1  \\ 0 & 0 \end{pmatrix}~,\qquad A_1 = \begin{pmatrix} 0 & 1  \\ 0 & 0 \end{pmatrix}~.
\end{equation}
As the MPS matrices satisfy 
\begin{equation}
A_0^2 = A_0, \quad A_1^2 = 0, \quad A_1A_0 = 0, \quad \text{and} \quad A_0A_1 = A_1,     
\end{equation}
this tensor has $\cV_j=\Span\{A_0,A_1\}$ for all $j$. It is left $1$-stable, for example with $Z=Y_0=A_0$ and $Y_1=0$. Therefore it satisfies the intersection property at lengths $k\geq 2$. Note that this tensor is non-normal and thus is an example of how stability generalizes injectivity. One can also check the intersection property directly by calculating $\mathcal{S}_k(\bA) = \Span \left\{ \ket{0\cdots 00},\ket{0\cdots01} \right \}$.

Now consider the transpose tensor $B_i = (A_i)^T$, which has $\mathcal{S}_k(\bB) = \Span\left\{ \ket{10\cdots 0},\ket{00\cdots 0} \right\}$ and satisfies the intersection property too. While $\bA$ is left $1$-stable, its transpose $\bB$ is \emph{right} $1$-stable. In general, the transpose of a left (right) stable tensor is right (left) stable, with matrices $Y_i^T$.

While the direct sum of two left (right) stable tensors is also left (right) stable, there is no guarantee of stability for the sum of a left stable tensor with a right stable tensor. In fact, the direct sum tensor $C_i = A_i \oplus B_i$ is neither left nor right stable, as it does not satisfy the intersection property. To see this, note that $\mathcal{S}_k(\bC) = \Span \left\{ \ket{0^{k}},\ket{0^{k-1}1},\ket{10^{k-1}} \right \}$, which satisfies
\begin{equation}
   (\mathcal{S}_k(\bC)\otimes \mathcal{H}) \cap (\mathcal{H}\otimes \mathcal{S}_k(\bC)) = \Span \left\{\ket{0^{k+1}}, \ket{0^{k}1},\ket{10^{k}}, \ket{10^{k-1}1} \right \} \equiv \mathcal{S}_{k+1}(\mathbf{D}) \neq \mathcal{S}_{k+1}(\bC)~, 
\end{equation}
where
\begin{equation}
    D_0 =\begin{pmatrix}
        0 & 0 & 0 & 0 \\
        1 & 1 & 0 & 0 \\
        0 & 0 & 1 & 0 \\
        0 & 0 & 0 & 0        
    \end{pmatrix} \ , \quad 
    D_1 =\begin{pmatrix}
        0 & 0 & 0 & 0 \\
        0 & 0 & 0 & 0 \\
        1 & 0 & 0 & 0 \\
        0 & 1 & 1 & 0        
    \end{pmatrix} \ .
\end{equation}
On the other hand, the new tensor $\mathbf{D}$ satisfies the usual intersection property.

If the OBC parent Hamiltonian of $\bC$ is frustration free, its ground space is $\mathcal{S}_{n}(\mathbf{D})$ since the intersection of the spaces $\mathcal{S}_k(\bC)$ results in $\mathcal{S}_{k+1}(\mathbf{D})$, that is:
\begin{equation}
    \ker H_n(\bC,\ell)=\bigcap_{i=0}^{n-\ell}\left(\cH^{\otimes i}\otimes\cS_\ell(\bC)\otimes\cH^{\otimes n-\ell-i}\right)=\bigcap_{i=0}^{n-(\ell+1)}\left(\cH^{\otimes i}\otimes\cS_{\ell+1}(\mathbf{D})\otimes\cH^{\otimes n-\ell-i}\right)=\cS_n(\mathbf{D})~.
\end{equation}

This example motivates us to define the following notion, which ensures that
$GS(H_n(\bA,\ell))=\mathcal{S}_n(\bB)$ in the frustration free case, for interaction lengths $\ell>k$:
\begin{definition}\label{defgenintprop}[Generalized intersection property] An MPS tensor $\bB \in\mathcal{H}\otimes \mathcal{M}_{D_B}$ is said to generalize the intersection property of the MPS tensor $\bA \in\mathcal{H}\otimes \mathcal{M}_{D_A}$ on $k$ sites if
\begin{equation}
    (\mathcal{S}_k(\bA)\otimes \mathcal{H}) \cap (\mathcal{H}\otimes \mathcal{S}_k(\bA)) = \mathcal{S}_{k+1}(\bB) \  ,   
\end{equation}
and $\bB$ satisfies the intersection property for all $j>k$.
\end{definition}

\section{Conclusion and outlook}

In this work we have discovered a condition that ensures that the parent Hamiltonian of an MPS has only the MPS as ground space. This condition, called stability, generalizes block injectivity. It is satisfied by important families of states like the W state, Dicke states, and domain wall superpositions.

Our investigation poses several questions for future study. One question concerns alternative formulations of the stability condition. With the current definition, it can be difficult to show that a tensor is not stable by ruling out the existence of matrices $Y_i$. The normality condition can be alternatively defined as the absence of an invariant subspace of the $A_i$ and the non-degeneracy of the highest eigenvalue of the transfer matrix \cite{PerezGarcia07}, and it would be interesting to similarly reformulate the stability condition in terms of transfer matrices. A related question is whether the ground spaces of stable MPS parent Hamiltonians can be determined from the transfer matrix spectrum, as they can in the block-injective case. A reformulation of stability might also lead to a version of the quantum Wielandt theorem \cite{Sanz10} for stability; that is, an upper bound on the stability length given that a tensor is stable at all.

Another question concerns the physics of MPS built out of stable tensors. Injective MPS display particular entanglement patterns such as exponentially decaying correlations. It would be interesting to study which physical features are shared by stable MPS.

We also leave for future work the study of the generalized intersection property of \cref{defgenintprop}. As demonstrated by the constructions of \cref{sec:examples-nonstable}, stability is a sufficient but not necessary condition for intersection. The AKLT tensor also exemplifies this: it satisfies intersection on $k=2$ sites yet is only $2$-stable, not $1$-stable. Indeed, a straightforward numerical check shows that the AKLT tensor has no $2$-site operator $O$ satisfying \cref{prop:Opmove}. It thus remains an open question to explain why certain normal (stable) tensors satisfy the intersection property at or below the injectivity (stability) length. It would be interesting to search for a property of the local MPS tensor that is equivalent to the intersection property or (more generally) to bounded ground state degeneracy of the parent Hamiltonians.

\section*{Acknowledgments}

This research has been funded in part by the European Union’s Horizon 2020 research and innovation program through Grant No.\ 863476 (ERC-CoG SEQUAM). J.G.R. acknowledges funding by the FWF Erwin Schrödinger Program (Grant DOI 10.55776/J4796). Research at the Perimeter Institute is supported in part by the Government of Canada through the Department of Innovation, Science and Economic Development and by the Province of Ontario through the Ministry of Colleges and Universities.


\printbibliography

\end{document}